\documentclass[11pt,journal]{IEEEtran}

\usepackage{amsfonts}
\usepackage{times}
\usepackage{latexsym}
\usepackage{amssymb}
\usepackage{amsmath}
\usepackage{cite}
\usepackage{verbatim}

\newcommand{\bydef}{\triangleq}
\newcommand{\tr}{{\it{tr}}}
\def\SNR{{\textsf{SNR}}}

\def\bydef{:=}

\def\bb0{{\mathbb{0}}}

\def\bydef{:=}
\def\ba{{\mathbf{a}}}
\def\bb{{\mathbf{b}}}

\def\bh{{\mathbf{h}}}

\def\bn{{\mathbf{n}}}

\def\br{{\mathbf{r}}}

\def\bv{{\mathbf{v}}}

\def\bx{{\mathbf{x}}}
\def\by{{\mathbf{y}}}
\def\bz{{\mathbf{z}}}
\def\b0{{\mathbf{0}}}

\def\bA{{\mathbf{A}}}

\def\bF{{\mathbf{F}}}
\def\bG{{\mathbf{G}}}
\def\bH{{\mathbf{H}}}
\def\bI{{\mathbf{I}}}

\def\bQ{{\mathbf{Q}}}

\def\bW{{\mathbf{W}}}

\def\bbC{{\mathbb{C}}}

\def\bbE{{\mathbb{E}}}

\def\bbN{{\mathbb{N}}}

\def\bbR{{\mathbb{R}}}

\def\bbZ{{\mathbb{Z}}}

\def\bydef{:=}

\def\sf0{{\mathsf{0}}}

\def\Nt{{N_t}}
\def\Nr{{N_r}}

\newcommand{\expeq}{\stackrel{.}{=}}
\newcommand{\expg}{\stackrel{.}{\ge}}
\newcommand{\expl}{\stackrel{.}{\le}}

\usepackage{graphicx}
\usepackage{amssymb}
\usepackage{amsfonts}
\usepackage{amsmath}
\usepackage{dsfont}

\begin{document}
\newtheorem{thm}{Theorem}
\newtheorem{lemma}[thm]{Lemma}
\newtheorem{rem}[thm]{Remark}
\newtheorem{exm}[thm]{Example}
\newtheorem{prop}[thm]{Proposition}
\newtheorem{defn}[thm]{Definition}
\newtheorem{cor}[thm]{Corollory}
\def\proof{\noindent\hspace{0em}{\itshape Proof: }}
\def\endproof{\hspace*{\fill}~\QED\par\endtrivlist\unskip}
\def\bh{{\mathbf{h}}}
\title{End-to-End Joint Antenna Selection Strategy
and Distributed Compress and Forward Strategy for Relay Channels}
\author{Rahul~Vaze and Robert W. Heath Jr. \\
The University of Texas at Austin \\
Department of Electrical and Computer Engineering \\
Wireless Networking and Communications Group \\
1 University Station C0803\\
Austin, TX 78712-0240\\
email: vaze@ece.utexas.edu, rheath@ece.utexas.edu
\thanks{This work was funded by DARPA through IT-MANET grant no. W911NF-07-1-0028.}}

\date{}
\maketitle
\noindent
\begin{abstract}

Multi-hop relay channels use multiple relay stages, each with multiple
relay nodes, to facilitate communication between a source and
destination. Previously, distributed space-time codes were proposed to maximize
the achievable diversity-multiplexing tradeoff, however, they fail to
achieve all the points of the optimal diversity-multiplexing tradeoff.
In the presence of a low-rate feedback link from the
destination to each relay stage and the source, this paper proposes
an end-to-end antenna selection (EEAS) strategy as an alternative to
distributed space-time codes.
The EEAS strategy
uses a subset of antennas of each relay stage for transmission of the source signal to the
destination with amplify and forwarding at each relay stage. The subsets are chosen
such that they maximize the end-to-end mutual information at the destination.
The EEAS strategy achieves the corner points of the optimal
diversity-multiplexing
tradeoff (corresponding to maximum diversity gain and maximum multiplexing
gain) and achieves better diversity gain at intermediate values of
multiplexing gain, versus the best known distributed space-time coding
strategies.
A distributed compress and forward (CF) strategy is also
proposed to achieve all points of
the optimal diversity-multiplexing tradeoff for a two-hop relay channel with
multiple relay nodes.

\end{abstract}

\section{Introduction}
Finding optimal transmission strategies for wireless ad-hoc networks in
terms of capacity, reliability, diversity-multiplexing (DM) tradeoff \cite{Zheng2003}, or delay, has been a long standing open
problem. The multi-hop relay channel is an important building block of
wireless ad-hoc networks. In a multi-hop relay channel, the source uses
multiple relay nodes to communicate with a single destination.
 An important first step in finding optimal transmission strategies for the wireless ad-hoc networks is to find optimal transmission strategies for the multi-hop relay channel.

In this paper, we focus on the design of transmission strategies to achieve
the optimal DM-tradeoff of the multi-hop relay channel.
The DM-tradeoff [1] characterizes the maximum achievable reliability (diversity gain)
for a given rate of increase of transmission rate (multiplexing gain), with increasing signal-to-noise
ratio (SNR). The DM-tradeoff curve is characterized by a set of points, where each point is a two-tuple
whose first coordinate is the multiplexing gain and the second coordinate is the maximum diversity gain
achievable at that multiplexing gain. We consider a multi-hop relay channel, where a source uses N -1
relay stages to communicate with its destination, and each relay stage is assumed to have one or more
relay nodes. Relay nodes are assumed to be full-duplex. Under these assumptions we find and characterize multi-hop relay strategies that achieve the DM-tradeoff curve (in the two hop case) or come close to the optimum DM-tradeoff curve while outperforming prior work (with more than two hops).

In prior work there have been many different transmit strategies proposed
to achieve the optimal DM-tradeoff of the multi-hop relay channel, such as
distributed space time block codes
(DSTBCs) \cite{Laneman2003, Laneman2004,
Nabar2004, Jing2004d, Jing2006a, Belfiore2007, Yiu2005,
Barbarossa2004, Damen2007, Oggier2006k, Jing2007, Jing2008,
Yang2007a, Sreeram2008, Vaze2008, Oggier2007b},
or relay selection \cite{Peters2007, Laneman2003, Laneman2004, Bletsas2006,
Zinan2005, Ibrahim2008, Caleb2007, Tanni2008}.
The best known DSTBCs \cite{Yang2007a, Sreeram2008} achieve the corner
points of the optimal DM-tradeoff of the multi-hop relay channel, corresponding
to the maximum diversity gain and maximum multiplexing gain,
however, fail to achieve the optimal
DM-tradeoff for intermediate values of multiplexing gain. Moreover, with DSTBCs \cite{Yang2007a, Sreeram2008} the
encoding and decoding complexity can be quite large.
Antenna selection (AS) or relay selection (RS) strategies have been
designed to achieve only the maximum diversity gain point of the optimal
DM-tradeoff when a small amount of feedback
is available from the destination, for a two-hop relay channel in
\cite{Peters2007, Laneman2003, Laneman2004, Bletsas2006, Zinan2005, Ibrahim2008, Caleb2007, Tanni2008}, and for a multi-hop relay channel in \cite{Vaze2008jsptocode}.
RS is also used for routing in multi-hop networks \cite{Park2003,Gui2007,Bohacek2008} to leverage path diversity gain. The primary advantages of AS and RS strategies over DSTBCs are that they require a minimal number of active antennas and reduce the encoding and decoding complexity compared to DSTBCs.
The only strategy that is known to achieve
all points of the optimal DM-tradeoff is the compress and forward (CF) strategy \cite{Yuksel2007}, but that is limited to a $2$-hop relay channel with a
single relay node.

In this paper we design an end-to-end antenna selection (EEAS)
strategy to maximize the achievable diversity gain for a given multiplexing gain
 in a multi-hop
relay channel. The EEAS strategy chooses a subset of
antennas from each relay stage that maximize the mutual information
at the destination \footnote{The proposed EEAS strategy is an extension of
the EEAS strategy proposed in \cite{Vaze2008jsptocode}, where only a
single antenna of each relay stage was used for transmission.}. The proposed
EEAS strategy is shown to achieve the corner points of the optimal
DM-tradeoff corresponding to maximum diversity gain and maximum
multiplexing gain. For intermediate values of multiplexing gains,
the achievable DM-tradeoff of the EEAS strategy does not meet with
an upper bound on the DM-tradeoff, but outperforms the achievable
DM-tradeoff of the best known DSTBCs \cite{Sreeram2008}. Other
advantages of the proposed EEAS strategy over DSTBCs \cite{Yang2007a,
Sreeram2008} include lower bit error rates due to less noise accumulation at
the destination, reduced decoding complexity and lesser latency.
We assume that the destination has the
channel state information (CSI) for all the channels in the receive mode.
Using the CSI, the destination performs subset selection, and
using a low rate feedback link feedbacks the index of the antennas to be used
by the source and each relay stage.

Even though our EEAS strategy performs better than the best known
DSTBCs \cite{Yang2007a, Sreeram2008},
it fails to achieve all points of the optimal DM-tradeoff.
To overcome this limitation, we propose a distributed CF strategy to achieve all points of the optimal DM-tradeoff of a $2$-hop relay channel with multiple relay nodes. Previously, the CF strategy of \cite{Cover1979} was shown to achieve
all points of the optimal DM-tradeoff of the $2$-hop relay channel with a
single relay node in \cite{Yuksel2007}.
The result of \cite{Yuksel2007}, however, does not extend for more
than one relay node.  With our
distributed CF strategy, each relay transmits a compressed version
of the received signal using Wyner-Ziv coding \cite{Wyner1976} without
decoding any other relay's message.
The destination first decodes the relay signals, and then uses the decoded
relay messages to decode the source message.

Our distributed strategy is a special case of the distributed CF strategy
proposed in \cite{Kramer2005}, where relays perform partial decoding of
other relay messages, and then use distributed compression to send their signals to
the destination. With partial decoding, the achievable rate
expression is quite complicated \cite{Kramer2005}, and it is hard to
compute the SNR exponent of the outage probability.
To simplify the achievable rate expression, we consider a special case of the CF strategy \cite{Kramer2005} where no relay decodes any other
relay's message. Consequently, the derivation for the SNR exponent
of the outage probability is simplified and we show that the special case of
CF strategy \cite{Kramer2005} is sufficient to achieve the optimal DM-tradeoff
for a $2$-hop relay channel with multiple relays.

{\it Organization:} The rest of the paper is organized as follows.
In Section \ref{sec:sys}, we describe the system model for the multi-hop
relay channel and summarize the key assumptions.
We review the diversity multiplexing (DM)-tradeoff for multiple antenna
channels in Section \ref{sec:dmtintro} and obtain an upper bound on the
DM-tradeoff of multi-hop relay channel.
In Section \ref{sec:full-dup} our EEAS strategy for the
multi-hop relay channel is described and its DM-tradeoff is computed.
In Section \ref{sec:2hop} we describe our distributed CF strategy and show
that it can achieve the optimal DM-tradeoff of $2$-hop relay channel with
any number of relay nodes.
Final conclusions are made in Section \ref{sec:conc}.

{\it Notation:}
We denote by ${\bA}$ a matrix, ${\bf a}$ a vector and $a_i$ the
$i^{th}$ element of ${\bf a}$. $\bA^{\dag}$ denotes the transpose conjugate of
matrix $\bA$.
The maximum and minimum eigenvalue of $\bA$ is denoted
 $\lambda_{max}(\bA)$ and $\lambda_{min}(\bA)$, respectively.
The determinant and trace of matrix  ${\bf A}$ is denoted by
$\det({\bA})$ and $\tr({\bA})$.
The field of real and complex numbers is denoted by $\bbR$ and $\bbC$,
respectively. The set of natural numbers is denoted by $\bbN$. The set $\{1, 2, \ldots n\}$ is denoted by
$[n], \ n \in \bbN$. The set $[n]/k$ denotes the set $\{1, 2, \ldots, k-1, k, \ldots n\}$, $k,\ n \ \in\bbN$. 
$[x]^{+}$ denotes $\max \{x,0\}$.
The space of $M\times N$ matrices with complex entries is denoted by
${\bbC}^{M\times N}$.
The Euclidean norm of a vector $\bf a$ is denoted by $|\ba|$.
The superscripts $^T, ^{\dag}$ represent the transpose and the transpose
conjugate.
The cardinality of a set ${\cal S}$ is denoted by $|{\cal S}|$.
The expectation of function $f(x)$ with respect to $x$
is denoted by ${\bbE}_{x}(f(x))$.
A circularly symmetric complex Gaussian random variable $x$ with zero mean and
variance $\sigma^2$
is denoted as $x \sim {\cal CN}(0,\sigma)$.
We use the symbol $\expeq$ to represent exponential equality i.e.,
let $f(x)$ be a
function of $x$, then  $f(x) \expeq x^a$ if $\lim_{x\rightarrow \infty}\frac{\log(f(x))}{\log x} = a$ and similarly $\expl$ and $\expg$ denote the exponential
less than or equal to and greater than or equal to relation, respectively.
To define a variable we use the symbol $\bydef$.

\section{System Model}
\label{sec:sys}
\begin{figure*}
\centering
\includegraphics[width= 7in]{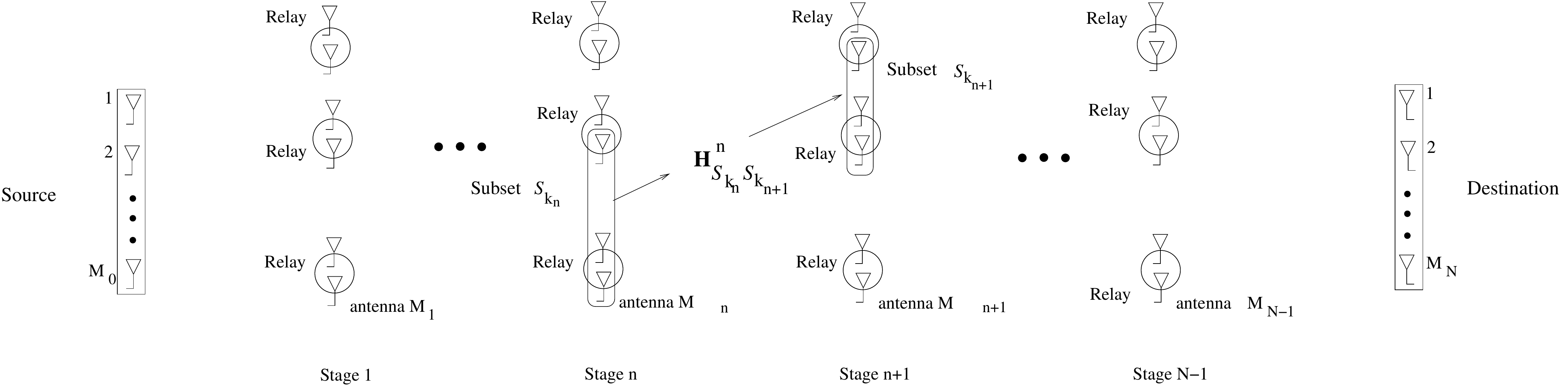}
\caption{System block diagram of a multi-hop relay channel with $N-1$ stages}
\label{blkdiag}
\end{figure*}

We consider a multi-hop relay channel where a source terminal with $M_0$
antennas wants to communicate with a destination terminal with
$M_N$ antennas via $N-1$ stages of relays as shown in Fig. \ref{blkdiag}.
The $n^{th}$ relay stage has $K_n$ relays and the $k^{th}$ relay
of $n^{th}$ stage has $M_{kn}$ antennas $n=1,2,\ldots,N-1$.
The total number of antennas in the
$n^{th}$ relay stage is $M_n \bydef \sum_{k=1}^{K_n}M_{kn}$.
In Section \ref{sec:2hop} we consider a $2$-hop relay channel with $K$ relay
nodes, where the $k^{th}$ relay has $m_k$ antennas and $\sum_{k=1}^K m_k= M_1$.
We assume that the relays do not generate their own data and each
relay stage has an average power constraint of $P$.
We assume that the relay nodes are synchronized at the frame level.
To keep the
relay functionality and relaying strategy simple we do not allow relay nodes to
cooperate among themselves.
For Section \ref{sec:full-dup} we assume that there is no direct path
between the source and the destination, but relax this assumption in Section
\ref{sec:2hop} for the $2$-hop relay channel. The absence of the direct path
is a reasonable
assumption for the case when relay stages are used for coverage
improvement and the signal strength on the direct path is very weak.
We also assume that relay stages are chosen in such a way that all the relay nodes of any two adjacent relay stages are connected to each other and there is no direct path between relay stage $n$
and $n+2$. This assumption is reasonable for the case when successive relay stages appear in increasing order of distance from the source towards the destination and any two relay nodes are chosen to lie in adjacent relay stages if they have sufficiently  good SNR between them. In any practical setting there will be interference received at any relay node of stage $n$ because of the signals transmitted from relay nodes of relay stage $0,\ldots,n-2$ and $n+2,\ldots,N-1$. Due to relatively large distances between non adjacent relay stages, however, this interference is quite small and we account for that in the additive noise term.
The system model is similar to the fully connected layered network with intra-layer links
\cite{Sreeram2008} and more general than the directed multi-hop relay channel model of
\cite{Yang2007a}.
We consider the full-duplex multi-hop relay channel,
where each relay node can transmit and receive at the
same time.

As shown in Fig. \ref{blkdiag}, the channel matrix between the
subset ${\cal S}_{k_n} \subset [M_n]$ of antennas of stage $n$ and the subset
${\cal S}_{k_{n+1}} \subset [M_{n+1}]$ of antennas of stage $n+1$ is denoted by
$\bH^n_{{\cal S}_{k_n}{\cal S}_{k_{n+1}}}$,
$ k_n=0,1,\ldots,$ $ {M_n}\choose {m}$, where
$|{\cal S}_{k_{n}}| =m \ \forall \ n$.
Stage $0$ represents the source and stage $N$ the destination.

In Section \ref{sec:2hop}, we only consider a $2$-hop relay channel and
denote the channel matrix between the source and $k^{th}$ relay by $\bH_{k}$
and between the $k^{th}$ relay and destination by $\bG_k$. The channel
between the source and destination is denoted by $\bH_{sd}$ and the channel
matrix between relay $k$ and relay $\ell$ by $\bF_{k\ell}$.

We assume that the CSI is known only at the destination and none of the relays have any CSI, i.e. the destination knows
$\bH^n_{{\cal S}_{k_n}{\cal S}_{k_{n+1}}}$, $ k_n=0,1,\ldots,$
${M_n}\choose {m}$, $ \ n =0,1,\ldots,N$. For Section \ref{sec:2hop}, we assume that
the destination knows $\bH_k, \bG_k, \bH_{sd}, \ \forall \ k$
and the $k^{th}$ relay node knows $\bH_{sd}, \bH_k$ and $\bG_k$
We assume that $\bH^n_{{\cal S}_{k_n}{\cal S}_{k_{n+1}}}, \bH_k, \bG_k, \bH_{sd}$ and $\bF_{k\ell}$ have independent and identically distributed (i.i.d.) ${\cal CN}(0,1)$ entries for all $n$ to model the channel as Rayleigh fading with
uncorrelated transmit and receive antennas.
We assume that all these
channels are frequency flat, block fading channels, where the channel coefficients remain constant in a block of time duration $T_c \ge N$ and change independently from block to block.


\section{Problem Formulation}
\label{sec:dmtintro}
We consider the design of transmission strategies to
achieve the DM-tradeoff of the multi-hop relay channel. In the next subsection
we briefly review the DM-tradeoff \cite{Zheng2003} for
point-to-point channels and obtain an upper bound on the DM-tradeoff of the
multi-hop relay channel.

 Review of the DM-Tradeoff:
\label{sec:dmt}
Following \cite{Zheng2003}, let ${\cal C}(\SNR)$ be a family of codes, one for
each $\SNR$. The multiplexing gain of
${\cal C}(\SNR)$ is $r$ if the data rate $R(\SNR)$ of ${\cal C}(\SNR)$ scales as $r$ with respect to
$\log \SNR$, i.e.
\[\lim_{\SNR\rightarrow \infty}\frac{R(\SNR)}{\log \SNR} =r.\]
Then the diversity gain $d(r)$ is defined as the rate of fall of probability of error $P_e$ of ${\cal
C}(\SNR)$ with respect to \SNR, i.e.
\[P_{e}(\SNR) \expeq \SNR^{-d(r)}.\]
The exponent $d(r)$ is called the  diversity gain at rate
$R=r\log\SNR$, and the curve joining $\left(r, d(r)\right)$ for
different values of $r$ characterizes the DM-tradeoff.
The DM-tradeoff for a point-to-point multi-antenna channel with
$N_t$ transmit and $N_r$ antennas has been computed in \cite{Zheng2003} by
first showing that $P_{e}(\SNR) \expeq P_{out}(r\log \SNR)$ and
then computing the exponent $d_{out}(r)$, where
\begin{equation}
\label{dmtmimo}
P_{out}(r\log \SNR)\expeq \SNR^{-d_{out}(r)}, \end{equation}
where $  d_{out}(r) = (\Nt-r)(\Nr-r)$,
for $r=0,1,\ldots,\min \{\Nt,\Nr\}$.

Next, we present an upper bound
on the DM-tradeoff of the multi-hop relay channel obtained in \cite{Yang2007a}.
\begin{lemma} \cite{Yang2007a}
\label{upbounddmt}
The DM-tradeoff curve of the multi-hop relay channel $\left(r,d(r)\right)$ is upper bounded by
 the piecewise linear function connecting the points
$\left(r,d^n(r)\right)$, $r=0,1,\ldots,\min\{M_n,M_{n+1}\}$ where
\[d^n(r) = (M_n-r)(M_{n+1}-r),\] for each $n=0,1,2,\ldots,N-1$.
\end{lemma}

The upper bound on the DM-tradeoff of multi-hop relay channel is obtained
by using the cut-set bound \cite{Cover2004} and allowing all relays in
each relay stage to cooperate.
Using the cut-set bound it follows that the mutual information between the
source and the destination cannot be more than the mutual information
between the source and any relay stage or between any two relay stages.
Moreover, by noting the fact that
mutual information between any two relays stages is upper bounded
by the maximum mutual information of a point-to-point MIMO channel
with $M_n$ transmit and $M_{n+1}$ receive antennas, $ n=0,1,\ldots,N-1$,
the result follows from (\ref{dmtmimo}).

In the next section we propose an EEAS strategy for the multi-hop relay
channel and compute its DM-tradeoff. We will show that the achievable
DM-tradeoff of the EEAS strategy meets the upper bound at $r=0$ and
$r=\min_{n=0,1,\ldots,N}M_n$.

\section{Joint End-to-End Multiple Antenna Selection Strategy}
\label{sec:full-dup} In this section we propose a joint end-to-end
multiple antenna selection strategy (JEEMAS)
for the multi-hop relay
channel and compute its DM-tradeoff. In the JEEMAS strategy,
a fixed number $(=m)$ of antennas are chosen from each relay
stage, to forward the
signal towards the destination using amplify and forward (AF).
Before introducing our JEEMAS strategy and
analyzing its DM-tradeoff, we need the following definitions and
Lemma \ref{lem:maxindepaths}.

\begin{defn} Let ${\cal S}_{k_n}$ be a subset of antennas of stage $n$, i.e.
${\cal S}_{k_{n}}\subset [M_n]$. Let $e^n_{{\cal S}_{k_n}{\cal S}_{k_{n+1}}}$ be
the edge joining the set of antennas ${\cal S}_{k_n}$ of stage $n$ to
the set of antennas ${\cal S}_{k_{n+1}}$  of stage $n+1$, where $|{\cal
S}_{k_n}|=m, \forall, n$. Then a path in a multi-hop relay channel is
defined as the sequence of edges $\left(e^0_{{\cal S}_{k_0}{\cal S}_{k_1}}, e^1_
{{\cal S}_{k_{1}}{\cal S}_{k_{2}}}
, \ldots, e^{N-1}_
{{\cal S}_{k_{N-1}}
{\cal S}_{k_{N}}}\right)$.
\end{defn}

\begin{defn}
Two paths $\left(e^0_{{\cal S}_{k_0}{\cal S}_{k_1}}, e^1_
{{\cal S}_{k_{1}}{\cal S}_{k_{2}}}
, \ldots, e^{N-1}_
{{\cal S}_{k_{N-1}}
{\cal S}_{k_{N}}}\right)$
and
$\left(e^0_{{\cal S}_{l_0}{\cal S}_{l_1}}, e^1_
{{\cal S}_{l_{1}}{\cal S}_{l_{2}}}
, \ldots, e^{N-1}_
{{\cal S}_{l_{N-1}}
{\cal S}_{l_{N}}}\right)$
 are called independent if
${\cal S}_{k_n} \cap {\cal S}_{l_n} = \phi, \ \forall \ n=0,1,\ldots,N$.
\end{defn}


In the next lemma we compute the maximum number of independent paths in
a multi-hop relay channel.
\begin{lemma}
\label{lem:maxindepaths}
The maximum number of independent paths in a multi-hop relay channel is
\[\alpha \bydef \min\left\{\left\lfloor\frac{M_n}{m}\right\rfloor\left\lfloor\frac{M_{n+1}}{m}
\right\rfloor\right\}, \ n=0,1,\ldots,N-1.\]
\end{lemma}
\begin{proof} Follows directly from Theorem 3 \cite{Vaze2008jsptocode} by replacing $M_n$ by
$\left\lfloor\frac{M_n}{m}\right\rfloor$.
\end{proof}

Now we are ready to describe our JEEMAS strategy for the full-duplex
multi-hop relay channel.
To transmit the signal from the source to the destination, a single path in a
multi-hop relay channel is used for communication. How to choose that path is described in the
following. Let the chosen path for the transmission be
$\left(e^0_{{\cal S}_{k^*_0}{\cal S}_{k^*_1}}, e^1_
{{\cal S}_{k^*_{1}}{\cal S}_{k^*_{2}}}
, \ldots, e^{N-1}_
{{\cal S}_{k^*_{N-1}}
{\cal   S}_{k^*_{N}}}\right)$. Then the signal is transmitted from the
${\cal S}_{k^*_0}^{*th}$ subset of antennas
of the source and is relayed through ${\cal S}_{k^*_n}^{th}$ subset of antennas of
relay stage $n, \ n=1,2,\ldots N-1$ and decoded by the ${\cal S}_{k^*_N}^{th}$
subset of antennas of the destination.
Each antenna on the chosen path uses an AF strategy to
forward the signal to the next relay stage,
i.e. each antenna of stage $n$ on the chosen path transmits the received signal after multiplying by $\mu_n$, where $\mu_n$ is chosen to satisfy an average power constraint $P$ across $m$ antennas of stage $n$.


Therefore with AF by each antenna subset on the chosen path,
the received signal at the ${\cal S}_{k^*_N}^{th}$ subset of antennas of the destination at time $t+N$ of a multi-hop relay channel is
\begin{align}
\label{rxsig}\nonumber
\br_{t+N} &= \prod_{n=0}^{N-1}\sqrt{\frac{P\mu_n}{m}}\bH^n_{{\cal S}_{k^*_n}
{\cal S}_{k^*_{n+1}}}\bx_t & \\ \nonumber
& \ \ +  \sum_{j=1}^{t-1}
\sqrt{\frac{P\gamma_j}{m}}f_j\left(\bH^n_{{\cal S}_{k^*_n}
{\cal S}_{k^*_{n+1}}}\right)\bx_{t-j} &\\
& \ \  + \underbrace{\sum_{m=1}^{N-1}\prod_{l=m}^{N-1}\sqrt{\mu_l}
q_l\left(\bH^n_{{\cal S}_{l^*_n}
{\cal S}_{l^*_{n+1}}}\right)\bv_{{\cal S}_{l_n^*}} +
\bv_{{\cal S}_{k_N^*}}}_{\bz_{t+N}}&,
\end{align}
where $f_j\left(\bH^n_{{\cal S}^*_{k_n}
{\cal S}^*_{k_{n+1}}}\right)$ and
$q_l\left(\bH^n_{{\cal S}^*_{k_n}
{\cal S}^*_{k_{n+1}}}\right)$ are functions of
channel coefficients $\bH^n_{{\cal S}^*_{k_n}
{\cal S}^*_{k_{n+1}}}$,  $\mu_n$ ensures
that the power constraint at each stage is met,
$\gamma_j$ is a function of $\mu_n$'s, $\bv_{{\cal S}_{l_n^*}}, n=1,2,\ldots, N$ is the complex Gaussian noise with
zero mean and unit variance added at stage $n$ and
$\mu_0=1$. Since the destination has the CSI, accumulated noise $\bz_{t+N}$ is
white and Gaussian distributed. From hereon in this paper we assume that the accumulated noise at the destination for all the multi-hop relay channels is white Gaussian distributed without explicitly mentioning it. Let $(\bW)^{-1}$
be the covariance matrix of $\bz_{t+N}$, then by multiplying
$\bW^{\frac{1}{2}}$ to the received signal we have
\begin{eqnarray}
\nonumber
r^{'}_{t+N} &=& \bW^{\frac{1}{2}}\prod_{n=0}^{N-1}\sqrt{\frac{P\mu_n}{m}}
\bH^n_{ {\cal S}_{k_n^*} {\cal S}_{k_{n+1}^*}}\bx_t \\ \nonumber
&& +
\bW^{\frac{1}{2}}\sum_{j=1}^{t-1}
\sqrt{\frac{\gamma_jP}{m}}f_j\left(\bH^n_{{\cal S}_{k^*_n}{\cal S}_{k^*_{n+1}}}\right)\bx_{t-j}\\\label{rxsigtrans}
&&+ \bz'_{t+N}
\end{eqnarray} where $\bz^{'}_{t+N}$ is a matrix with ${\cal CN}(0,1)$ entries.
Note that $\bW$ is a function of channel coefficients $\bH^n_{{\cal S}^*_n{\cal S}^*_{n+1}}$.

We propose to use successive decoding at the destination with the
JEEMAS strategy, similar to \cite{Vaze2008jsptocode}. With successive
decoding, the destination tries to decode only $\bx_t$ at time $t+N,
\ t=1,2,\ldots,T, \ T\le T_c$ assuming that all the symbols
$\bx_1,\bx_2,\ldots,\bx_{t-1}$ have been decoded correctly. Assuming
that at time $t+N$ all the symbols $\bx_1,\bx_2,\ldots,\bx_{t-1}$
have been decoded correctly, the received signal (\ref{rxsigtrans})
can be written as
\begin{equation}
\label{succdec}
\br^{eq}_{t+N} = \bW^{\frac{1}{2}}\prod_{n=0}^{N-1}
\sqrt{\frac{P\mu_n}{m}}\bH^n_{ {\cal S}_{k_n^*}{\cal S}_{k_{n+1}^*}}\bx_t + \bz'_{t+N},
\end{equation} since the channel coefficients $\bH^n_{{\cal S}^*_n{\cal S}^*_{n+1}}$ are known at the destination.
Let the  probability of error in decoding  $\bx_t$ from (\ref{succdec}) be
$P_t$, then the probability of error $P_e$ in decoding $\bx_1,\bx_2,\ldots,
\bx_T$ from  (\ref{rxsigtrans}) with successive decoding $P_e$ is
\begin{eqnarray}
\label{errprobsuccdec}
P_e &\le& 1-\prod_{t=1}^T(1-P_t) \\\nonumber
 & \expl&P_t  \ \text{for any} \ t, \ t=1,\ldots,T,
 \end{eqnarray}
where the last equality follows from \cite{Vaze2008jsptocode}.

From (\ref{succdec}) it is clear that $P_t$ is the same for any $t, \
t=1,2,\ldots,T$, since the channel coefficients $\bH^n_{ {\cal S}_{k_n^*}{\cal S}_{k_{n+1}^*}}$ do not change for $T\le T_c$ time instants.
Therefore without loss of generality we compute an
upper bound on $P_1$ to upper bound $P_e$. Next, we describe our
JEEMAS strategy and compute an upper bound on $P_1$ of the
JEEMAS strategy to evaluate its DM-tradeoff. Let $\SNR \bydef \frac{P}{m}\prod_{n=0}^{N-1}\mu_n$.
Let $\Pi_{k_0}^{k_N} = \prod_{n=0}^{N-1}\bH^n_{ {\cal S}_{k_n}
{\cal S}_{k_{n+1}}}$, then the mutual information of
path $\left(e^0_{{\cal S}_{k_0}{\cal S}_{k_1}},
e^1_{{\cal S}_{k_1}{\cal S}_{k_2}}, \ldots, e^{N-1}_{{\cal S}_{k_{N-1}}{\cal
S}_{k_N}}\right)$ is
\begin{align*} 
&M.I.\left(\bW^{\frac{1}{2}}
\Pi_{k_0}^{k_N}\right) \bydef&  \\ 
& \ \ \ \log\det\left(\bI_{m} + \SNR \  \bW^{\frac{1}{2}}
\Pi_{k_0}^{k_N}\Pi_{k_0}^{k_N \dag}\bW^{\frac{1}{2}\dag}\right).
\end{align*}

Then the JEEMAS strategy chooses the path
that maximizes the mutual information at the destination, i.e. it chooses
path $(e^0_{{\cal S}_{k_0^*}{\cal S}_{k_1^*}},
e^1_{{\cal S}_{k_1^*}{\cal S}_{k_2^*}}, \ldots, e^{N-1}_{{\cal S}_{k_{N-1}^*}{\cal S}_{k_N^*}})$, if
\begin{align*}
\label{maxcriteria}
&{\cal S}_{k_0^*}, {\cal S}_{k_1^*}, {\cal S}_{k_{N-1}^*}, {\cal S}_{k_N^*} = & \\
& \ \ \ \ arg \max_{\begin{array}{c}
{\cal S}_{k_n} \subset [M_n], \\
\ n\in\{0,1,\ldots,N\}
\end{array}}M.I.\left(\bW^{\frac{1}{2}}
\Pi_{k_0}^{k_N}\right). &\end{align*}
Thus defining $\Pi^* = \prod_{n=0}^{N-1}
\bH^n_{{\cal S}_{k_n^*}{\cal S}_{k_{n+1}^*}}$, the mutual information
of the chosen path is
\begin{align*}
&M.I.\left(\bW^{\frac{1}{2}}
\Pi^*\right) \bydef& \\
&  \ \  \log\det\left(\bI_{m} + \SNR\bW^{\frac{1}{2}}
\Pi^*\Pi^{*\dag}\bW^{\frac{1}{2}\dag}\right).&\end{align*}
Since we assumed that the destination of the multi-hop relay channel
has CSI for all the channels in the receive mode, this optimization
can be done at the destination and using a feedback link, the source
and each relay stage can be informed about the index of antennas to
use for transmission. Next, we evaluate the DM-tradeoff of the
JEEMAS strategy by finding the exponent of the outage probability
(\ref{succdec}).

From \cite{Zheng2003} we know that $P_1\expeq P_{out}(r\log \SNR)$, where  $P_{out}(r\log \SNR)$ is the outage probability of (\ref{succdec}). Therefore it is sufficient to compute an upper bound on the outage probability of (\ref{succdec}) to upper bound $P_e$.
With the proposed EEAS strategy, the outage probability of
(\ref{succdec}) can be
 written as
\begin{align}
\nonumber
P_{out}(r\log \SNR)=
P\left(M.I.\left(\bW^{\frac{1}{2}}
\Pi^*\right)\le r\log \SNR\right).
\end{align}
From \cite{Yang2007a, Sreeram2008} $\bW^{\frac{1}{2}}$ can be dropped from the
DM-tradeoff analysis without changing the outage exponent,
since $ \lambda_{max}\left(\bW^{\frac{1}{2}}\right) \expeq  \lambda_{max}\left(\bW^{\frac{1}{2}}\right)\expeq \SNR^0$ \cite{Yang2007a}, i.e. the maximum or the minimum eigenvalue of $\bW^{\frac{1}{2}}$ do not scale with $\SNR$.
Thus,
\begin{align}
P_{out}(r\log \SNR) \expeq  
P\left(M.I.\left(
\Pi^*\right)
\le
r\log\SNR\right).
\end{align}
We first compute the DM-tradeoff of the JEEMAS strategy for the case when there
exists $\alpha_n$ such that
$M_n = \alpha_n m, \ \forall \ n=0,1,\ldots,N$, and then for the general case.

If $M_n = \alpha_n m, \ \forall \ n=0,1,\ldots,N$, then by Lemma \ref{lem:maxindepaths}, the
total number of independent paths in a multi-hop relay channel is
$\kappa\bydef\min_{n=0,1,\ldots,N-1}\{\alpha_n\alpha_{n+1}\}$. Thus,
\begin{eqnarray*}
P_{out}(r\log \SNR) \le 
\left(
P
\left(M.I.\left(\Pi^{k_N}_{k_0}\right)
\le r\log\SNR \right)
\right)
^{\kappa},
\end{eqnarray*}
since from (\ref{maxcriteria}) 
$M.I.\left(\Pi^{*}\right)
\ge M.I.\left(\Pi^{k_N}_{k_0}\right)$
for any $\Pi_{k_0}^{k_N}$.

From \cite{Yang2007a}
\begin{align}
 \label{dmtexpsym}
P\left(M.I.\left(\Pi^{k_N}_{k_0}\right)
\le
r\log\SNR\right)\expeq \SNR^{-d^N_{m}(r)},&
\end{align}
where
\begin{eqnarray*}
d_m^N(r) &=& \frac{(m-r)(m+1-r)}{2} \\ &&+ \frac{a(r)}{2}\left((a(r)-1)N+2b(r)\right),
\end{eqnarray*}
where $a(r) \bydef \left \lfloor\frac{m-r}{N}\right \rfloor $, and
$b(r) \bydef (m-r)  \ mod  \ N$.
Thus, $P_{out}(r\log \SNR) \le \SNR^{-\kappa d_m^N(r)} $ and
the DM-tradeoff of the JEEMAS strategy is given by
\begin{eqnarray*}d(r) = \kappa d_m^N(r).\end{eqnarray*}

For the general case when $M_n \ne \alpha_n m, \ \forall \
n=0,1,\ldots,N$, let $M_n = \alpha_n m + \beta_n, \ \beta_n \le m$,
for some $\alpha_n$ and $\beta_n$. Then partition the multi-hop
relay channel into two parts, the first partition ${\cal P}_1$
containing $\alpha_n m$ antennas of each stage, such that the chosen
set of antennas by the JEEMAS strategy ${\cal S}_{k_n^*} \subset {\cal
P}_1, \ \forall \ n$, and the second partition ${\cal P}_2$
containing the rest $\beta_n$ antennas of each stage. By reordering
the index of antennas, without loss of generality, let ${\cal P}_1$
contain antennas $1$ to $\alpha_n m$ of each relay stage and ${\cal
P}_2$ contain antennas $\alpha_n m +1$ to $\alpha_n m +\beta_n$ of
stage $n$. Recall that the JEEMAS strategy chooses those $m$
antennas of each stage that have the maximum mutual information at
the destination. Thus,
\begin{align}
\nonumber
&P_{out}(r\log \SNR) = &\\ \nonumber
&\ \ \ \ P\left(\max_{{\cal S}_{k_n} \subset [M_n]}
M.I.\left(\Pi^{k_N}_{k_0}\right)\le
r\log\SNR\right) & \\ \nonumber
&\le  P\left(\max_{{\cal S}_{k_n} \subset [\alpha_n m]} 
M.I.\left(\Pi^{k_N}_{k_0}\right)\le
r\log\SNR\right., & \\
&\ \ \ \ \ \  \ \ \ \ \ \ \ \ \ \  \ \ \ \left.M.I.\left(\Pi_{last}\right)
\le
r\log\SNR\right) &
\label{poutindpaths}
,
\end{align}
where $\Pi_{last} = \prod_{n=0}^N\bH^{n}_{{\cal S}^{last}_{ n}{\cal S}^ {last}_{n+1}}$, and
$\bH^{n}_{{\cal S}^{last}_{ n}{\cal S}^ {last}_{n+1}}$ is the $m \times m$ channel matrix
between $M_n-m+1$ to $M_n$ antennas of stage $n$ and $M_{n+1}-m+1$ to $M_{n+1}$ antennas of stage
$n+1$.  Note that the channel coefficients in $\Pi_{last} $ are not independent of the channel
coefficients in $\Pi_{k_0}^{k_N}, \ {\cal S}_{k_n} \subset [\alpha_n m]$, and therefore we cannot write $P_{out}(r\log \SNR) $ as the product of
 \[P\left(\max_{{\cal S}_{k_n} \subset [\alpha_n m]}M.I.\left(\Pi^{k_N}_{k_0}\right)\le r\log \SNR\right)\] and 
 \[P \left(M.I.\left(\Pi_{last}\right)
\le
r\log\SNR\right).\]
To circumvent this problem, let $\Pi_{{\cal P}_2} =
\bH^{0}_{{\cal S}^{last}_{0}\beta_1}\bH^1_{\beta_1 {\cal S}^ {last}_{n+1}} \ldots \bH^{N-1}_{{\cal S}^{last}_{N-1}\beta_{N}}$,
where $\bH^n_{{\cal S}^{last}_{n}\beta_{n+1}}$ is the channel matrix between
the last $m$ antennas of stage $n$ and the last $\beta_{n+1}$ antennas of stage $n+1$ of partition ${\cal P}_2$, and $\bH^n_{\beta_{n}{\cal S}^{last}_{n+1}}$
is the channel matrix between the last $\beta_n$ antennas of stage $n$ and the
last $m$ antennas of stage $n+1$ of partition ${\cal P}_2$.
Basically we pick $m$ and $\beta_n$ antennas alternatively from each stage,
such that the channel coefficients in $\Pi_{{\cal P}_2} $ are independent of channel coefficients in $
\Pi_{k_0}^{k_N}, \ {\cal S}_{k_n} \subset [\alpha_n m]$.
Note that $\Pi_{{\cal P}_2}$ uses a subset of antennas of $\Pi_{last}$, and
since outage probability decreases by using more antennas of each
stage \footnote{Use of more antennas increases the mutual information
of the channel, and consequently reduces the outage probability.}, from (\ref{poutindpaths})
\begin{align*}
&P_{out}(r\log \SNR)\le  & \\
& \ \ \ \ P\left(\max_{{\cal S}_{k_n} \subset [\alpha_n m]} M.I.\left(\Pi^{k_N}_{k_0}\right)\le
r\log\SNR, \right. &  \\ 
& \ \ \ \ \ \ \ \ \  \ \ \ \ \  \ \ \ \ \ \left. M.I.\left(\Pi_{{\cal P}_2}\right) \le r\log\SNR\right).&
\end{align*}
Since the channel coefficients in $\Pi_{{\cal P}_2} $ are independent of
the channel coefficients of $\Pi_{k_0}^{k_N}, \ {\cal S}_{k_n} \subset [\alpha_n m]$,
\begin{align*}
&P_{out}(r\log \SNR)\le  & \\
& \ \ \ \ P\left(\max_{{\cal S}_{k_n} \subset [\alpha_n m]} M.I.\left(\Pi^{k_N}_{k_0}\right)\le
r\log\SNR, \right) \times &  \\ 
&\ \ \ \ P\left( M.I.\left(\Pi_{{\cal P}_2}\right) \le r\log\SNR\right).&
\end{align*}

Therefore, 
\begin{align*}
&P_{out}(r\log \SNR)\le  & \\
& \ \ \ \ P\left( M.I.\left(\Pi^{k_N}_{k_0}\right)\le
r\log\SNR, \right)^{\kappa} \times&  \\ 
& \ \ \ \  P\left( M.I.\left(\Pi_{{\cal P}_2}\right) \le r\log\SNR\right)&
\end{align*}
since the number of independent paths in partition ${\cal P}_1$ are
$\kappa$.

From \cite{Yang2007a},
$P\left( M.I.\left(\Pi_{{\cal P}_2}\right) \le r\log\SNR\right)
 = \SNR^{-(d_{m,\beta_1, m,\ldots,m,\beta_{N}}(r))},$
where
\begin{eqnarray*}d^N_{m,\beta_1,m,\ldots,m,\beta_{N}}(r)&=&
\sum_{k=r+1}^{\beta_{min}} 1-k \\ 
&&+ \min_{n=1,\ldots,N} \left\lfloor \frac{\sum_{l=0}^n{\hat \beta}_l-k}{n}\right\rfloor,
\end{eqnarray*}
$r=0,1,\ldots,\min\{\beta_1, \ldots, \beta_N, m\}$, 
where $\beta_{min} \bydef \min\{\beta_1, \beta_3, \ldots, \beta_{N}\}$ and
$\left\{{\hat \beta}_0, {\hat \beta}_1, \ldots, {\hat \beta}_N\right\}$
is the non-decreasing ordered version of $\left\{m,\beta_1, m\ldots, m,\beta_{N}\right\}$, ${\hat \beta}_0 \le {\hat \beta}_1 \le \ldots \le {\hat \beta}_N$.
Thus,
\begin{eqnarray*}
P_{out}(r\log \SNR)
&\le & \SNR^{-\left(\kappa d_m^N(r) + d^N_{m,\beta_1,m\ldots,m,\beta_{N}}(r)\right)}.
\end{eqnarray*}
Therefore, using (\ref{dmtexpsym}), the DM-tradeoff of the JEEMAS strategy is
\begin{eqnarray}
\label{dmtexp}d(r) = \kappa d_m^N(r) + \left[d^N_{m,\beta_1,m\ldots,\beta_{N-1},m}(r)\right]^{+},
\end{eqnarray}
$r=0,1,\ldots,\min_{n=0,1,\ldots,N} \{M_n\}$.

Recall that in the JEEMAS strategy the design parameter is $m$, the number
of antennas to use from each stage.
To obtain the best lower bound on the DM-tradeoff of JEEMAS strategy
one needs to find out the optimal value of $m$.
From (\ref{dmtexp}), it follows that using a single antenna $m=1$,
maximum diversity gain point can be achieved.
Similarly, choosing $m=\min_{n=0,\ldots,N} M_n$,
the maximum multiplexing gain point can also be achieved.
For intermediate values of $r$, however, it is not apriori clear what value of
$m$ maximizes the diversity gain. After tedious computations it turns out that
choosing $m=\min_{n=0,\ldots,N} M_n$ provides with the best
achievable DM-tradeoff for $r>0$. Thus, we propose a hybrid JEEMAS strategy,
where for $r=0$ use $m=1$, and for $r>0$ use $m=\min_{n=0,\ldots,N} M_n$.
Our approach is similar to \cite{Sreeram2008}, where for each $r$
an optimal partition of the multi-hop relay channel is found by solving
an optimization problem. We compare the achievable DM-tradeoff
of our hybrid JEEMAS strategy and the strategy of \cite{Sreeram2008} for
$M_0=2,M_1=4,M_2=2$ and $M_0=3,M_1=5,M_2=3$ in Figs. \ref{dmt242} and \ref{dmt353}.
\begin{figure}
\centering
\includegraphics[width= 3.5in]{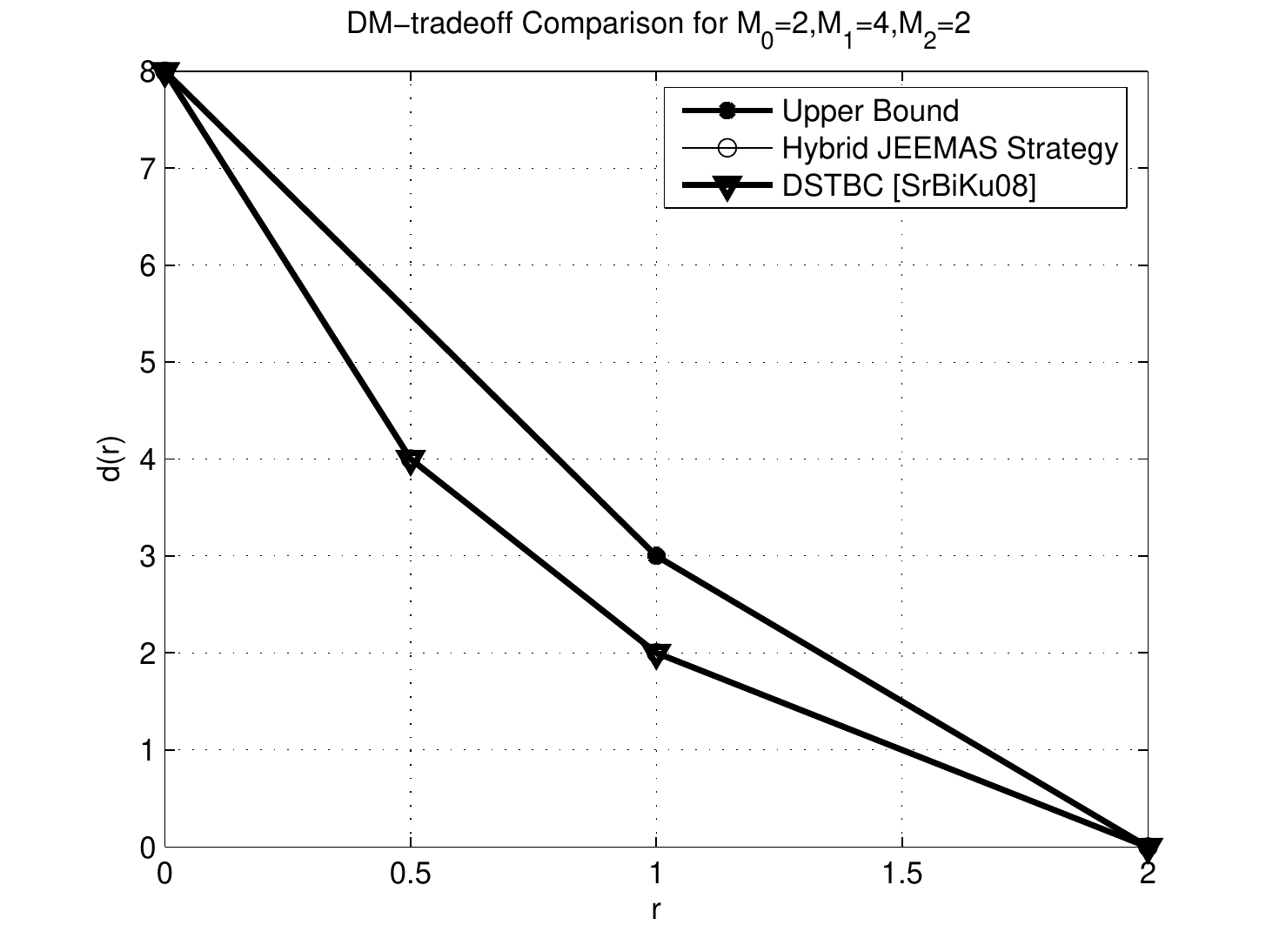}
\caption{DM-tradeoff comparison of hybrid JEEMAS with the strategy of \cite{Sreeram2008}}
\label{dmt242}
\end{figure}
\begin{figure}
\centering
\includegraphics[width= 3.5in]{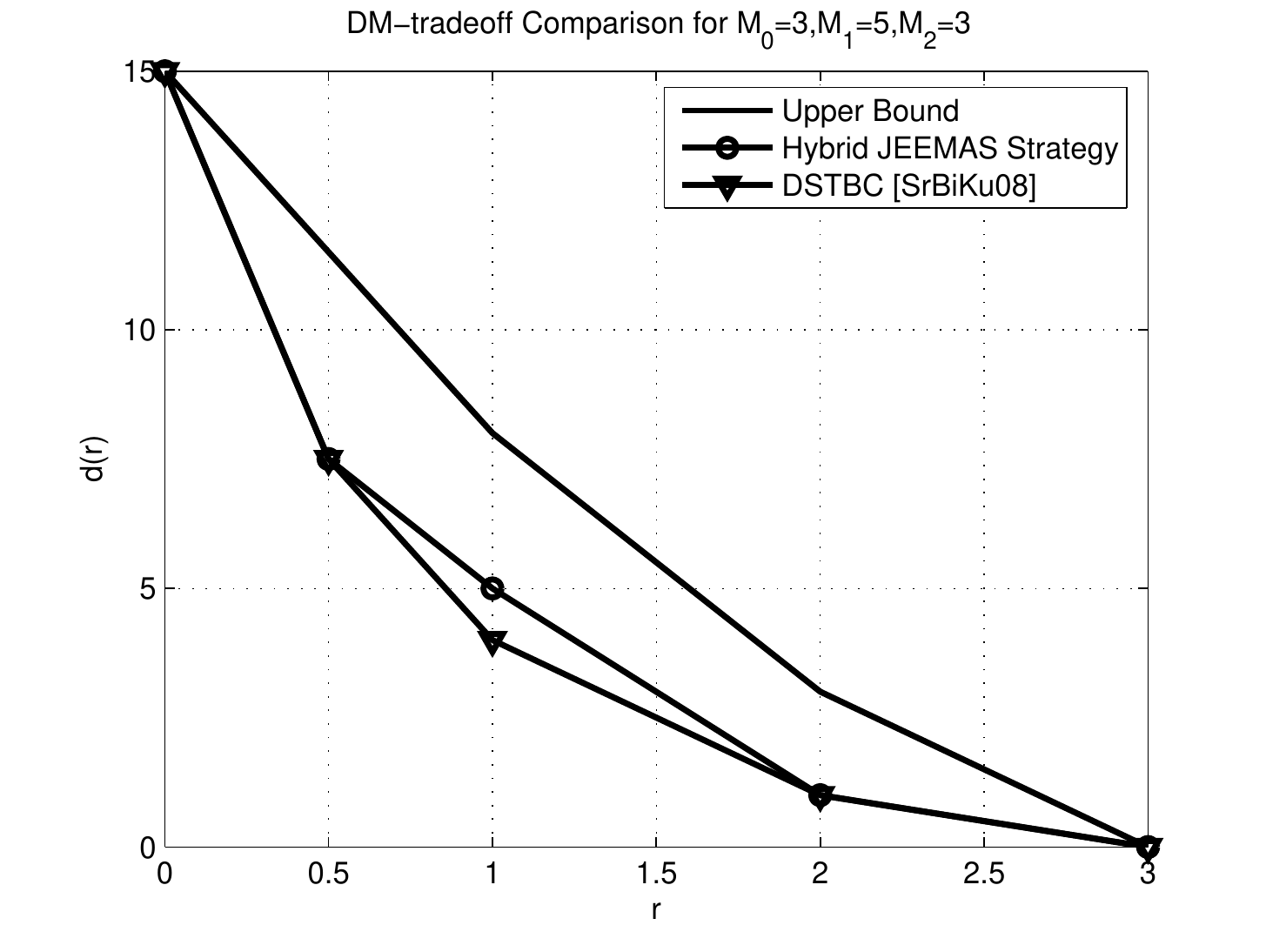}
\caption{DM-tradeoff comparison of hybrid JEEMAS with the strategy of \cite{Sreeram2008} }
\label{dmt353}
\end{figure}

For the case when $\beta_n=0, \ \forall \ n$, the achievable DM-tradeoff of
our hybrid JEEMAS strategy matches with that of the partitioning strategy of
\cite{Sreeram2008}. For the case when $\beta_n\ne0, \ \forall \ n$,
however, it is difficult to compare the hybrid
JEEMAS strategy with the strategy of \cite{Sreeram2008} in terms of
achievable DM-tradeoff, since an optimization problem has to be solved for
the strategy of \cite{Sreeram2008}. For a particular example of $N=2, M_0=3,
M_1=5, M_2=3$ the hybrid JEEMAS strategy outperforms the strategy of
\cite{Sreeram2008} as illustrated in Fig. \ref{dmt353}.
Moreover, in \cite{Sreeram2008} a new partition is required for each $r$,
in contrast to our strategy, which has only two modes of operation,
one for $r=0$ and the other for $r>0$.

The following remarks are in order.
\begin{rem} Recall that we  assumed that $|S_{k_n}| = m$; i.e.
equal number of antennas are selected at each relay stage. The justification
of this assumption is as follows.
Let us assume that $M_n, \ n=0,1,\ldots,N$ antennas are used from each
relay stage.
Now assume that all relay stages are using the same number of antennas $M_n=m, \ \forall \ n, \ n \ne l$,
except $l$, which is using $k$ antennas, $M_l=k$, and $m\ne k$.
Using (\ref{dmtexp}), it can be shown that the achievable DM-tradeoff with
$M_n=m, \forall \ n, n \ne l$, and $M_l=k$ is a subset of the union of the
achievable DM-tradeoff's with using $M_n=m, \forall \ n$ (all relay stages using $m$ antennas), and
$M_n=k, \forall \ n$ (all relay stages using $k$ antennas).
Thus, it is sufficient to consider
same number of antennas from each relay stage.
It turns out, however, that different values of $m$ provide with
different achievable DM-tradeoff's because of the different number of
independent paths in the multi-hop relay channel.
To optimize over all possible values of $m$ we keep $m$ as a variable,
and choose $m$ to obtain the best achievable DM-tradeoff.

\end{rem}
\begin{rem}
\label{dmtpp}
Using the DM-tradeoff analysis of the JEEMAS strategy, we can obtain the DM-tradeoff of an antenna
selection strategy for the point to point MIMO channel by considering a
multi-hop relay channel with $N=0$, $M_t$ transmit, and $M_r$ receive antennas
such that
($M_t\ge M_r$). Surprisingly we could not find this result in the literature
and provide it here for completeness sake.
Let $M_t = \alpha M_r + \beta$, and the transmitter uses $M_r$ antennas
out of $M_t$ antennas that have maximum mutual information at the destination,
then the DM-tradeoff  is given by
\[d(r) = \alpha(M_r-r)(M_r-r) + \left[(\beta-r)(M_r-r)\right]^{+},\] 
$r=0,1,\ldots,M_r$. 
The proof follows directly
from (\ref{dmtexp}).
\end{rem}

\begin{rem}
\label{rem:csi} CSI Requirement:
With the proposed hybrid JEEMAS strategy, the destination needs to
feedback the index of the path with the maximum mutual information to
the source and each stage. Recall from the derivation of
the achievable DM-tradeoff of the JEEMAS strategy that only $\kappa$ paths
in a multi-hop relay channel are independent, and control
the achievable DM-tradeoff for $\beta_n=0, \ \forall \ n$. Thus, the destination
only needs to feedback the index of the best path among $\kappa$ independent paths with the maximum mutual information. Consequently the destination only needs
to know CSI for $\kappa$ paths. For the case when $\beta_n\ne0, \ \forall \ n$,
we need to consider one more path from partition ${\cal P}_2$ corresponding
to $m$ and $\beta_n$ antennas of alternate relay stages.
Thus, the CSI overhead is moderate for the proposed EEAS strategy.
\end{rem}

\begin{rem} Feedback Overhead:
As explained in Remark \ref{rem:csi}, to obtain the achievable DM-tradeoff
of the hybrid JEEMAS strategy it is sufficient to consider any one set of
$\kappa$ or $\kappa +1$ independent paths.
Let the destination choose a particular set $S$ of
$\kappa+1$ independent paths.
Then each relay node knows on which of the paths of $S$ it lies and depending
on the index of the element of $S$ from the destination, it knows whether to transmit or remain silent. Thus, only $\log_2(\kappa+1)$
bits of feedback is required from the destination to the source and each
stage. Therefore the feedback overhead with the proposed EEAS strategy is
quite small and can be realized with a very low rate feedback link.

\end{rem}

{\it Discussion:} In this section we proposed a hybrid JEEMAS
strategy that has two modes of operation, one for $r=0$, where it
uses a single antenna of each stage, and the other for $r>0$, that
uses $\min_{n=0,\ldots,N} M_n$ antennas of each stage. The proposed
strategy is shown to achieve both the corner points of the optimal
DM-tradeoff curve, corresponding to the maximum diversity gain and
the maximum multiplexing gain. For intermediate values of
multiplexing gain, the diversity gain of our strategy is quite close
to that of the upper bound. Even though our strategy does not meet the upper bound, we
show that it outperforms the best known DSTBC strategy
\cite{Sreeram2008}, with smaller complexity and possess several
advantages over DSTBCs as described in \cite{Vaze2008jsptocode}. In the next
section we propose a distributed CF strategy to achieve the optimal
DM-tradeoff of the $2$-hop relay channel.

\section{Distributed CF Strategy for $2$-Hop Relay Channel}
\label{sec:2hop} In this section we consider  a $2$-hop
relay channel with multiple relay nodes
in the presence of a direct path between the source
and the destination.
For this $2$-hop relay channel we propose a distributed compress and forward (CF)
strategy to achieve the optimal DM-tradeoff.
The signal model for this section is as follows. We
consider a $2$-hop relay channel with $K$ relay nodes, where the
$k^{th}$ relay has $m_k$ antennas, and $\sum_{k=1}^{K}m_k=M_1$. The source
and destination are assumed to have $M_0$ and $M_2$ antennas, respectively.
We assume that the source and each relay have an average power constraint of
$P$ \footnote{Different transmit power constraints do not change the
DM-tradeoff.}.
Let the signal transmitted from the source be $\bx$, and from the relay node $k$ be
$\bx_k$, respectively.
Then,
\begin{eqnarray}
\nonumber
\by &=& \sqrt\frac{{P}}{M_0}\bH_{sd}\bx + \sum_{k=1}^K\sqrt\frac{{P}}{m_k}
\bG_k\bx_k + \bn, \\
\by_{k} &=& \sqrt\frac{{P}}{M_0}\bH_{k}\bx + \sum_{\ell=1, k\ne \ell}^K
\sqrt\frac{{P}}{m_{\ell}}\bF_{k\ell}\bx_{\ell}
+ \bn_k, \label{rxsigdmtfd}
\end{eqnarray}
where $\by$ is the received signal at the destination, and $\by_k$ is the
signal received at relay $k$.

Previously in \cite{Yuksel2007}, the CF strategy of \cite{Cover1979} has been
shown to achieve the optimal DM-tradeoff of a $2$-hop relay channel with a single relay node
($K=1$) in the presence of direct path between the source and the destination.
The result of \cite{Yuksel2007}, however, does not generalize to the case of
$2$-hop relay channel with multiple relay nodes. The problem with multiple relay
nodes is unsolved, since how multiple relay nodes should cooperate among themselves to
help the destination decode
the source message is hard to characterize.
A compress and forward (CF)  strategy for a $2$-hop relay channel with multiple relay nodes has been
proposed in \cite{Kramer2005}, which involves
partial decoding of other relays messages at each relay and
transmission of correlated
information from different relay nodes to the destination using
distributed source coding. The achievable rate expression obtained in
\cite{Kramer2005}, however, is quite complicated and cannot be computed easily
in closed form.

The
achievable rate expression of the CF strategy \cite{Kramer2005} is complicated because each relay node
partially decodes all other relay messages. Partial decoding
introduces auxillary random variables which are hard to optimize
over. To allow analytical tractability, we simplify the strategy of
\cite{Kramer2005} as follows. In our strategy each relay compresses
the received signal from the source using Wyner-Ziv coding similar
to \cite{Kramer2005}, but without any partial decoding of any other
relay's message. The compressed message is then transmitted to the
destination using the strategy of transmitting correlated messages
over a multiple access channel \cite{Cover1980}. Our strategy is a
special case of CF strategy \cite{Kramer2005}, since in our case the
relays perform no partial decoding. Consequently our strategy leads
to a smaller achievable rate
 compared to \cite{Kramer2005}. The biggest advantage of our strategy, however,
is its easily computable achievable rate expression and its sufficiency in achieving the
optimal DM-tradeoff as shown in the sequel. We refer to our strategy as distributed CF
from hereon in the paper. Even though the relays
do not perform any partial decoding in the distributed CF strategy,
in the sequel we show that
they still provide the destination with enough information
about the source message to achieve the optimal DM-tradeoff.
Before describing our distributed CF strategy and showing its optimality in achieving the
optimal DM-tradeoff, we present an upper bound on the DM-tradeoff of the $2$-hop relay channel.
\begin{figure*}
\begin{equation*}L_{s} = \det\left(\frac{P}{M_0}\bH_{s}^{d}\bH_{s}^{d\dag} +
\left[\begin{array}{cccc}
\bI_{M_2} & 0 & 0 & 0 \\
0 & \left({\hat N}_1+1\right)\bI_{m_1} & 0 & 0\\
0 & 0 & \ddots & 0  \\
0 & 0 & 0 & \left({\hat N}_K+1\right)\bI_{m_K} \end{array}\right]\right),
\end{equation*}
\end{figure*}
\begin{lemma} \cite{Yang2007a}
\label{dmtupboundfd} The DM-tradeoff of a two-way relay channel is upper
bounded by
\begin{eqnarray*}
d(r) &\le& \min\{(M_0-r)(M_1+M_2-r), \\
&&(M_0+M_1-r)(M_2-r)\},
\end{eqnarray*}
$r=0,1,\ldots,\min\{M_0, M_1+M_2, M_0+M_1, M_2\}$.
\end{lemma}
\begin{proof}
Let us assume that all the relay nodes and the destination
are co-located and can cooperate perfectly. This assumption can only improve
$d(r)$. In this case, the communication model from the source to destination
is a point to point MIMO channel with $M_0$ transmit antennas and $M_1+M_2$
receive antennas. The DM-tradeoff of this MIMO channel is $(M_0-r)(M_1+M_2-r)$, and since this point to point MIMO channel is better than our original $2$-hop
relay channel, $d(r) \le (M_0-r)(M_1+M_2-r)$.
Next, we assume that the source is co-located with all the relay nodes and can
cooperate perfectly for transmission to the destination.
This setting is equivalent to a MIMO channel with $M_0+M_1$ transmit and
$M_2$ receive antenna with DM-tradeoff $(M_0+M_1-r)(M_2-r)$. Again,
this point to point MIMO channel is better than our original $2$-hop relay
channel and hence $d(r) \le (M_0+M_1-r)(M_2-r)$,
which completes the proof.
\end{proof}

To achieve this upper bound we propose the following distributed CF strategy.
Let the rate of transmission from source to destination be $R$. Then
the source generates $2^{nR}$ independent and identically
distributed $x^n$ according to distribution $p(x^n) =
\prod_{i=1}^np(x_{i})$. Label them $x(w), \ w \in [2^{nR}]$. The
codebook generation, the relay compression and transmission remains the
same as in \cite{Kramer2005}, expect that no relay node decodes any other
relay's codewords, i.e. no partial decoding at any relay node.
Relay node $k$ generates $2^{nR_{k}}$ independent
and identically distributed $x_k^n$ according to distribution
$p(x_k^n) = \prod_{i=1}^np(x_{ki})$ and labels them $x_k(s), \ s\in
[2^{nR_k}]$, and for each $x_k(s)$ generates $2^{n{\hat R}}$
${\hat y}_k$'s, each with probability $p({\hat y}_k|x_k(s)) =
\prod_{i=1}^np({\hat y}_{ki}|x_{ki}(s)) $. Label these ${\hat
y}_k(z_k|s), s \in [2^{nR_k}]$ and $z_k \in [2^{n{\hat R}_k}]$ and
randomly partition the set $ [2^{n{\hat R}_k}]$ into $2^{nR_k}$
cells $S_{s}, \ s\in [2^{nR_k}]$.

{\it Encoding:} A Block Markov encoding \cite{Cover1979} together with Wyner-Ziv coding \cite{Wyner1976} is used by each relay.
Let in block $i$ the message to send from the source be $w_i$,
then the source sends $x(w_i)$. Let the signal received by relay $k$ in block
$i$ be $y_k(i)$. Then $y_k(i)$ is compressed to ${\hat y}_k(z_{ik})$ using
Wyner-Ziv coding \cite{Wyner1976} where correlation among
$y_1,\ldots,y_K$ is exploited. Then relay $k$ determines the cell
index $s_{ik}$ in which $z_{ik}$ lies and transmits $x_k(s_{ik})$ in block $i+1$.
We consider transmission of $B$ blocks of $n$ symbols each from the
source in which $B-1$ messages will be sent.
Each message is chosen from $w \in [2^{nR}]$.
Thus, as $B \rightarrow \infty$, for fixed $n$,
rate $R\left(\frac{B-1}{B}\right)$ is arbitrarily close to $R$ \cite{Cover1979}.
In the first block, the relay has no information about $s_{0k}$ necessary for
compression. In this case, however, any good sequence allows each relay to
start block Markov encoding \cite{Cover1979}. In the last block, the source
is silent and only the relays transmit to destination.

{\it Decoding:} Backward decoding is employed at the destination.
At the end of block $i$, the codeword sent by source in block $i-1$ is
decoded. At the end of block $i$, the destination first decodes $x_k$ for each
$k$ by looking for a jointly typical $x_k(s_{ik})$ and $y_i$. If $R_k
\le I(\bx_k;\by | \bx_{[K]/ k})$, $x_k(s_{ik})$ can be decoding reliably.
Next, given
that $x_k$'s have been decoded correctly for each $k$, the destination
tries to find a set ${\cal L}$ of $z_1, \ldots, z_K$ such that
$\left( x_1(s_1), \ldots, x_K(s_K), {\hat y_1}(z_1|s_1), \ldots,
{\hat y_K}(z_K|s_K), y\right)$ is jointly typical. The destination declares
that $z_1, \ldots, z_K$ were the correctly sent codewords if
$(z_1, \ldots, z_K)\in
\left(S_{s_1} \times S_{s_2} \times \ldots \times
S_{s_K}\right) \cap {\cal L}$. After decoding $x_1(s_1), \ldots, x_K(s_K)$ and
$z_1, \ldots, z_K$ the destination decodes ${\hat w}$ if
$\left(x(w), x_1(s_1), \ldots, x_K(s_K), {\hat y_1}(z_1|s_1), \ldots,
{\hat y_K}(z_K|s_K), y\right)$ is jointly typical.
 With this distributed CF strategy,
\begin{eqnarray}
\label{cfratefd}
R &\le& I(\bx;\by, \hat{\by}_1, \ldots, \hat{\by}_K| \bx_1, \ldots, \bx_K)
\end{eqnarray}
is achievable with the joint probability distribution
\begin{eqnarray*}
p(x)\left[\prod_{k=1}^Kp(x_k)p({\hat y}_k|x_k,y_k)\right]\times \\ p(y_1. \ldots, y_K, y|
x,x_1,\ldots,x_K),
\end{eqnarray*}
subject to
\begin{eqnarray}
\nonumber
I(\hat{\by}_{\cal T}; \by_{\cal T}  | \bx_{[K]} {\hat \by}_{{\cal T}^C} \by)
+ \sum_{t\in {\cal T}} I(\hat{\by}_{t}; \bx_{[K] / {t}} | \bx_{t}) \\\label{compconstraintfd}
\le I(\bx_{{\cal T}};\by|\bx_{{\cal T}^C}), \ \forall \ {\cal T} \subseteq [K],
\end{eqnarray}
where $\by_{\cal T}$, $\hat{\by}_{\cal T}$ are vectors with elements
$\by_{t},\ \hat{\by}_{t}, \ t \in {\cal T},\ {\cal T} \subseteq [K]$,
respectively,
$\bx_{[K]}$ is the vector containing $\bx_1, \bx_2, \ldots, \bx_K$, and
${\cal T}^C$ is the complement of ${\cal T}$, where ${\cal T} \subseteq [K]$.
For more detailed error probability analysis we refer the reader to \cite{Kramer2005}.
In the next Theorem
we compute the outage exponents for (\ref{cfratefd}) and show that
they match with the exponents of the upper bound.

\begin{thm} CF strategy achieves the DM-tradeoff upper bound
(Lemma \ref{dmtupboundfd}).
\end{thm}
\begin{proof}
To prove the Theorem we will compute the achievable DM-tradeoff of the CF
strategy (\ref{cfratefd}) and show that it matches with the upper bound.

To compute the achievable rates subject to the compression rate
constraints for the signal model (\ref{rxsigdmtfd}), we fix ${\hat
\by}_k = \by_k + \bn_{qr}$, where $\bn_{qk}$ is $m_k \times 1$ vector with
covariance matrix ${\hat N}_k\bI_{m_k}$. Also, we choose $\bx$,
and $\bx_k$ to be  complex Gaussian with covariance
matrices $\frac{P}{M_0}\bI_{M_0}$,  and
$\frac{P}{m_k}\bI_{m_k}$, and independent of each other, respectively. Next, we compute the various mutual information
expressions to derive the achievable DM-tradeoff of the CF strategy.
By the definition of the mutual information
\begin{align*}
&I(\bx;\by, \hat{\by}_1,\ldots, \hat{\by}_K | \bx_1, \ldots, \bx_K) = & \\
& \ \ \ \ \ \ \ \ \ h(\by, \hat{\by}_1,\ldots, \hat{\by}_K | \bx_1, \ldots, \bx_K)&\\
& \ \ \ \ \ \ \ \ \ - h(\by, \hat{\by}_1,\ldots, \hat{\by}_K | \bx, \bx_1, \ldots, \bx_K). & \end{align*}
From (\ref{rxsigdmtfd}),
\begin{eqnarray}
\label{ents-rd}
h(\by, \hat{\by}_1,\ldots, \hat{\by}_K | \bx_1, \ldots, \bx_K) =
\log L_{s}, \end{eqnarray}
 where $L_s$ is defined on the top of the page and  $\bH_{s}^{d} = \left[\bH_{sd} \ \bH_1 \ldots \bH_K\right]^T$.
From (\ref{rxsigdmtfd}),
\begin{small}
\begin{align*}
&h(\by, \hat{\by}_1,\ldots, \hat{\by}_K | \bx, \bx_1, \ldots, \bx_K) =& \\
& \log \det \left(\left[\begin{array}{cccc}
\bI_{M_2} & 0 & 0 & 0 \\
0 & \left({\hat N}_1+1\right)\bI_{m_1} & 0 & 0\\
0 & 0 & \ddots & 0  \\
0 & 0 & 0 & \left({\hat N}_K+1\right)\bI_{m_K} \end{array}\right]
\right),& \end{align*}
\end{small}
which implies
\begin{align}\nonumber
&I(\bx;\by, \hat{\by}_1,\ldots, \hat{\by}_K | \bx_1, \ldots, \bx_K) = & \\ \label{mis-rd}
& \ \ \ \ \log\frac{ L_{s}}{({\hat N}_1+1)^{m_1}({\hat N}_2+1)^{m_2}\ldots({\hat N}_K+1)^{m_K}}.&\end{align}

Next, we compute the values of ${\hat N}_k$'s that satisfy the
compression rate constraints (\ref{compconstraintfd}). Note that in
(\ref{compconstraintfd}), we need to satisfy the constraints for each
subset ${\cal T} \subseteq [K]$.
Towards that end, first we consider the subsets ${\cal T}$ of the form
${\cal T} = \{k\}, \ k=1,2,\ldots,K$, and obtain the lower bound on the quantization noise
${\hat N}_k$ needed
to satisfy (\ref{compconstraintfd}), that is not proportional to $P$
for each $k$. It is important to note that ${\hat N}_k$ should not be
proportional to $P$, otherwise from (\ref{mis-rd}) it can be concluded that our
distributed CF strategy cannot achieve the optimal DM-tradeoff.
In the sequel we will point out how to obtain ${\hat N}_k$
satisfying (\ref{compconstraintfd}) for all subsets of $[K]$.

\begin{figure*}
\begin{equation*}
L_{s[K] / k} = \det\left(\frac{P}{M_0}\bH_{k}\bH_{k}^{\dag}+
\sum_{\ell =1, \ \ell \ne k}^K\frac{P}{m_{\ell}}\bF_{\ell k}\bF_{\ell k}^{\dag}+ ({\hat N}_k+1)\bI_{m_k}\right)
\end{equation*}
\end{figure*}

\begin{figure*}
\begin{equation*}
L_{s{\hat k}} = {\det \left(\left[\begin{array}{cc}
\left({\hat N}_k+1\right)\bI_{m_k} & 0 \\
0 & \bI_{M_2}\end{array}\right] + \frac{P}{M_0}[\bH_{k} \ \bH_{sd}]^T[\bH_{k}^{\dag} \ \bH_{sd}^{\dag}]\right) }
\end{equation*}
\end{figure*}

For ${\cal T}= \{k\}$, from (\ref{compconstraintfd}),
for each relay $k$, we need to satisfy
\begin{eqnarray} \nonumber
I(\hat{\by}_{k}; \by_{k}  | \bx_{[K]} {\hat \by}_{[K] / k} \by)
+ I(\hat{\by}_{k}; \bx_{[K] / {k}} | \bx_{k})
\\ \label{compnoiseeff}\le I(\bx_{k};\by|\bx_{[K]/ k}).
\end{eqnarray}

By definition
\begin{align}\nonumber
&I(\bx_k;\by|\bx_{[K]/ k}) = h(\by|\bx_{[K]/ k}) - h(\by|\bx_k \bx_{[K]/ k}), &\\ \nonumber
 & \ \ \ \ = \log\underbrace{\det\left(\frac{P}{M_0}\bH_{sd}\bH_{sd}^{\dag}+
\frac{P}{m_k}\bG_k\bG_k^{\dag}+\bI_{M_2}\right)}_{L_{skd}} & \\\label{noise1}
&\ \ \ \ \ \ \ - \log\underbrace{\det\left(\frac{P}{M_0}\bH_{sd}\bH_{sd}^{\dag}+
\bI_{M_2}\right)}_{L_{sd}} \ \text{using} \ (\ref{rxsigdmtfd}).&
\end{align}

Similarly,
\begin{align}\nonumber
&I(\hat{\by}_{k}; \bx_{[K] / {k}} | \bx_{k})  = h(\hat{\by}_{k}| \bx_{k}) - h(\hat{\by}_{k}| \bx_{[K] / {k}}  \bx_{k}), &\\ \label{noise2}
&=\log L_{s[K] / k}  - \log\underbrace{\det\left(\frac{P}{M_0}\bH_{k}\bH_{k}^{\dag}+ ({\hat N}_k+1)\bI_{m_k}\right)}_{L_{sk}},
&\end{align}
where $L_{s[K] / k} $ is defined on the top of the page.

\begin{align}\nonumber
&I(\hat{\by}_{k}; \by_{k}  | \bx_{[K]} {\hat \by}_{[K] / k} \by)
= h(\hat{\by}_{k}, \by | \bx_{[K]} {\hat \by}_{[K] / k} ) & \\ \nonumber
& \ \ \ \ \ \  \ \ \ \ \ \ \ \ \ \ \ \ \ \ \ \ \ \ \ \ \ \  - h(\by | \bx_{[K]} {\hat \by}_{[K] / k} ) - h(\hat{\by}_{k}| \by_{k}), & \\ \nonumber
&= \log L_{s{\hat k}} -
\log\underbrace{\det
\left(
\frac{P}{M_0}
\bH_{sd}\bH_{sd}^{\dag} +\bI_{M_2}
\right)
}_{L_{sd}}
&  \\\label{noise3}
&  \ \ - \log {\hat N}_k^{m_k},&
\end{align}
where $L_{s{\hat k}} $ is defined on the top of the page.

%
From (\ref{noise1}, \ref{noise2}, \ref{noise3}), to satisfy the compression rate constraints (\ref{compnoiseeff})
we need
\begin{equation}
\label{effcompnoise}
{\hat N}_k^{m_k} \ge \frac{L_{s[K]/k} L_{s{\hat k}}}{L_{skd}L_{sk}}.
\end{equation}
Note that both sides of (\ref{effcompnoise}) are functions
of ${\hat N}_k$, however, the resulting ${\hat N}_k$ is not a
function of $P$ or SNR similar to \cite{Yuksel2007}.
Recall that we have only considered the
subsets of $[K]$ of the form ${\cal T} =\{k\}$. For the rest of the
subsets also, we can show that the quantization noise ${\hat N}_k$
required to satisfy (\ref{compconstraintfd}) is not proportional to
$P$. The analysis follows similarly and is deleted for the sake
of brevity. Thus, to satisfy (\ref{compconstraintfd}), we can take
the maximum of the ${\hat N}_k$ required for each subset ${\cal
T}\subseteq [K]$ and use that to analyze the DM-tradeoff. Let the
maximum ${\hat N}_k$ required to satisfy (\ref{compconstraintfd}) be
${\hat N}_{max, k}$. Since ${\hat N}_k$ for each subset ${\cal
T}\subseteq [K]$ is not proportional to $P$, ${\hat N}_{max, k}$ is
also not proportional to $P$.

Then, using (\ref{cfratefd}) and (\ref{mis-rd}), we can
compute the outage probability of the distributed CF as follows.
From \cite{Zheng2003},
to compute $d(r)$, it is sufficient to find the negative of
the exponent of the $\SNR$ of outage probability at the destination,
where outage probability $P_{out}(r\log \SNR)$, is defined as
\begin{eqnarray*}
P_{out}(r\log \SNR) &=& P(R \le r\log \SNR).
\end{eqnarray*}
From (\ref{cfratefd}) and (\ref{mis-rd}),
\begin{equation}
R  = \log \frac{L_{s}}{({\hat N}_{max, 1}+1)^{m_1}\ldots({\hat N}_{max,  K}+1)^{m_K}}.
\end{equation}
Let $L_d \bydef \log\det\left(\frac{P}{M_0}\bH_{sd}\bH_{sd}^{\dag}+
\sum_{k=1}^M\frac{P}{m_k}\bG_k\bG_k^{\dag}+\bI_{M_2}\right)$.
Then choose $l_k \in \bbZ$ such that
\begin{equation}
\label{univnoise}
{\hat N}_{max,  k}
\le l_k\left(\left(\frac{L_{s}}{L_{d}}\right)^{1/{M_1}} + 1\right)
, \ \forall \ k.
\end{equation}
It is possible to choose $l_k$'s that satisfy (\ref{univnoise}),
since ${\hat N}_{max,  K}$ is not proportional to $P$.

Then
\begin{align*}
&P_{out}(r\log \SNR) =&  \\ 
& \ \ \ P\left(\log
\frac{L_{s}}
{\prod_{k=1}^K l_k\left(\left(\frac{L_{s}}{L_{d}}\right)^{1/{M_1}} + 1
\right)^{m_k}}  \le r\log \SNR\right),& \\
&  = P\left(\log \frac{L_{s}}
{\left(\left(
\frac{L_{s}}{L_{d}}
\right)^{1/{M_1}} + 1
\right)^{M_1}\prod_{k=1}^K l_k
}\le r\log \SNR\right).&
\end{align*}

\begin{align*}
&P_{out}(k\log \SNR)
\expeq&  \\ 
& \ \ 
P\left( \frac{L_{s}}
{
\left(\left(
\frac{L_{s}}{L_{d}}
\right)^{1/{M_1}} + 1
\right)^{M_1}
}\le \prod_{k=1}^K l_k \SNR^r\right),&\\
&=  P\left(
\frac{ \left(L_{s}\right)^{\frac{1}{M_1}} \left(L_{d}\right)^{\frac{1}{M_1}}  }
{ \left(L_{s}\right)^{\frac{1}{M_1}} +  \left(L_{d}\right)^{\frac{1}{M_1}}  }
\le \prod_{k=1}^K l_k^{\frac{1}{M_1}} \SNR^{\frac{r}{M_1}}\right),&\\
&=  P\left(
\frac{ \left(L_{s}\right)^{\frac{1}{M_1}} \left(L_{d}\right)^{\frac{1}{M_1}}  }
{ \left(L_{s}\right)^{\frac{1}{M_1}} +  \left(L_{d}\right)^{\frac{1}{M_1}}  }
\le \SNR^{\frac{r}{M_1}}\right),&
\end{align*}
where the last equality follows since multiplying $\SNR$ by constant does
not change the DM-tradeoff.

From here on we follow \cite{Yuksel2007} to compute the exponent of the
$P_{out}(r\log \SNR)$.

Let
\begin{equation*}L_{sl} = \det\left(\frac{P}{M_0}\bH_{s}^{d}\bH_{s}^{d\dag} +
\left[\begin{array}{ccccc}
\bI_{M_2}& 0  & 0 & 0\\
0 & \bI_{m_1} & 0  & 0 \\
0 & 0 & \ddots & 0 \\
0& 0 &  0 & \bI_{m_K}
\end{array}\right]\right).
\end{equation*}
Then, from (\ref{ents-rd}), $L_{sl} \le L_{s}$,
therefore using Lemma 2 \cite{Yuksel2007}, it follows that
\begin{eqnarray*}
P_{out}(r\log \SNR)
&\le & P\left( \left(L_{sl}\right)^{\frac{1}{M_1}} \le \SNR^{\frac{r}{M_1}}\right) \\
 &&+ \ P\left( \left(L_{d}\right)^{\frac{1}{M_1}} \le \SNR^{\frac{r}{M_1}}\right) , \\
&=& P\left( L_{sl} \le \SNR^{r}\right)\\
&& + \ P\left( L_{d} \le \SNR^{r}\right), \\
& \bydef & \SNR^{-d_1(r)} + \SNR^{-d_2(r)}.
\end{eqnarray*}
Therefore, to lower bound the DM-tradeoff we need to find out the
outage exponents $d_1(r)$ and $d_2(r)$ of $L_{sl}$ and
$L_{s}$. Notice that, however, $\log \left(L_{sl}\right)$ is the mutual
information between the source and the destination by choosing the
covariance matrix to be $\frac{P}{M_0}\bI_{M_0}$\footnote{$P$ taking
the role of $\SNR$.}, and allowing all the relays and the
destination to cooperate perfectly. From \cite{Zheng2003}, choice of
$\frac{P}{M_0}\bI_{M_0}$ as the covariance matrix does not change
the optimal DM-tradeoff, therefore, $d_1(r) = (M_0-r)(M_1+M_2-r)$.
Similar argument holds for $\log\left(L_{d}\right)$, by noting that
$\log\left(L_{d}\right)$ is
the mutual information between the source and the destination if all
the relays and the source were co-located and could cooperate
perfectly, while using covariance matrix $\bQ$, where
\[\bQ = \left[\begin{array}{cccc}\frac{P}{M_0}\bI_{M_0} & 0 & 0 & 0 \\
0 & \frac{P}{m_1}\bI_{m_1} & 0 & 0 \\
0 & 0& \ddots & 0 \\
0 & 0 & 0 & \frac{P}{m_K}\bI_{m_K}\end{array}\right].\]
Thus, $d_2(r) =
(M_0+M_1-r)(M_2-r)$. Thus, the achievable DM-tradeoff with CF
strategy meets the upper bound (Lemma \ref{dmtupboundfd}).
\end{proof}
{\it Discussion:}
In this section we proposed a simplified version of the distributed CF
strategy of \cite{Kramer2005} and showed that it can achieve the
optimal DM-tradeoff for the $2$-hop relay channel for any number of relays.
In our distributed CF strategy, each relay uses Wyner-Ziv coding to
compress the received signal without any partial decoding of other relay
messages. After compression, each relay
transmits the message to the destination using the
strategy for multiple access channel with correlated messages \cite{Cover1980},  since the relay compressed messages are correlated with each other.
Even though the achievable rate with our strategy is smaller than
the one obtained in \cite{Kramer2005} (because of no partial decoding at any
relay), we show that it is sufficient to achieve the optimal DM-tradeoff.
We prove the result by showing that the exponent of the outage probability
of our strategy matches with the upper bound on the optimal DM-tradeoff,
without requiring the compression noise constraints to be proportional to
the SNR.

Generalizing our distributed
CF strategy is possible for more than $2$-hop relay channel, however,
computing the exponents of the outage probability  of achievable rate and
compression rate constraints is a non-trivial problem.

\section{Conclusions}
\label{sec:conc}
In this paper we considered the problem of achieving the optimal DM-tradeoff
of the multi-hop relay channel. First, we proposed an antenna selection
strategy called JEEMAS, where a subset of antennas of each relay stage is chosen for
transmission that has the maximum mutual information at the destination. We
showed that the JEEMAS strategy can achieve the maximum diversity gain and
the maximum multiplexing gain in a multi-hop relay channel. Then we compared
the DM-tradeoff performance of the JEEMAS strategy with the best known
DSTBC strategy \cite{Sreeram2008}. We observed that the DM-tradeoff of the
JEEMAS is better than the DSTBCs \cite{Sreeram2008},
expect for the case when the number of antennas at each stage are divisible by the minimum of the antennas
across all relay stages, in which case the DM-tradeoffs of JEEMAS and DSTBCs
\cite{Sreeram2008} match.

Next, we proposed a distributed CF strategy for the $2$-hop relay channel
with multiple relay nodes and showed that it achieves the optimal DM-tradeoff.
Our distributed CF strategy is a special case of the strategy proposed in
\cite{Kramer2005}, where the specializations are done to allow analytical
tractability. We showed that if each relay transmits a compressed
version of the received signal using Wyner-Ziv coding, it
is sufficient to achieve the optimal DM-tradeoff. Our distributed CF strategy
can be extended to more than $2$-hop relay channels, however, computing the
outage probability exponents is a non-trivial problem.

\bibliographystyle{IEEEtran}
\bibliography{IEEEabrv,Research}

\end{document}